    \documentclass[final,floating,letterpaper]{siamltex1213}

    \usepackage{url}
    \usepackage{verbatim}
    \usepackage[titletoc]{appendix}
    \usepackage{graphicx}
    \usepackage{xcolor}
    \usepackage{hyperref}
    \usepackage{bm}
	\usepackage{amsmath}
	\usepackage{amsfonts}
	\usepackage{graphicx}
    \usepackage{nicefrac}
    \usepackage{longtable}
    \usepackage{color}
    \usepackage{graphicx, amssymb,graphics}
    \usepackage{algorithmicx,algorithm}
    \usepackage{paralist}
    \usepackage{float}

    \usepackage{epstopdf}
\makeatletter 
\let\myorg@bibitem\bibitem
\def\bibitem#1#2\par{%
	\@ifundefined{bibitem@#1}{%
		\myorg@bibitem{#1}#2\par
	}{%
		\begingroup
		\color{\csname bibitem@#1\endcsname}%
		\myorg@bibitem{#1}#2\par
		\endgroup
	}%
}

\newtheorem{rem}{Remark}[section]
\newtheorem{thm}{Theorem}[section]

\newcommand{\reff}[1]{{\rm (\ref{#1})}}
	\newcommand\be {\begin{equation}}
	\newcommand\ee {\end{equation}}

\newcommand{\R}{\mathbb{R}}            
\newcommand{\bD}{\bm{D}}
\newcommand{\bE}{\bm{E}}
\newcommand{\ve}{\varepsilon}          

\numberwithin{equation}{section}

\graphicspath{{./Figs/}}

\title{A Maxwell-Amp\`{e}re Nernst-Planck Framework for Modeling Charge Dynamics
}

\date{\today}

\begin{document}

	\author{
Zhonghua Qiao\thanks{Department of Applied Mathematics and Research Institute for Smart Energy, The Hong Kong Polytechnic University, Hung Hom, Hong Kong. (zhonghua.qiao@polyu.edu.hk)},~
\and
Zhenli Xu\thanks{School of Mathematical Sciences, MOE-LSC, and CMA-Shanghai, Shanghai Jiao Tong University, Shanghai 200240, China. (xuzl@sjtu.edu.cn)},~
\and
Qian Yin\thanks{School of Mathematical Sciences, MOE-LSC, and CMA-Shanghai, Shanghai Jiao Tong University, Shanghai 200240, China. (sjtu\_yinq@sjtu.edu.cn)},~
\and
Shenggao Zhou\thanks{School of Mathematical Sciences, MOE-LSC, and CMA-Shanghai, Shanghai Jiao Tong University, Shanghai 200240, China. (sgzhou@sjtu.edu.cn)}.
}

 \maketitle

\begin{abstract}	
Understanding the properties of charge dynamics is crucial to many practical applications, such as electrochemical energy devices and transmembrane ion channels. This work proposes a Maxwell-Amp\`{e}re Nernst-Planck (MANP) framework for the description of charge dynamics. The MANP model with a curl-free condition on the electric displacement is shown to be energy dissipative with respect to a convex free-energy functional, and demonstrated to be equivalent to the Poisson-Nernst-Planck model. By the energy dissipation law, the steady state of the MANP model reproduces the charge conserving Poisson--Boltzmann (PB) theory, providing an alternative energy stable approach to study the PB theory. {  In order to achieve the curl-free condition, a companion local curl-free relaxation algorithm, which is shown to naturally preserve the discrete Gauss's law and converge robustly with linear computational complexity, is developed for the MANP model.} One of the main advantages of our development is that it can efficiently deal with space-dependent permittivity instead of solving the variable-coefficient Poisson's equation. Many-body effects such as ionic steric effects and Coulomb correlations can be incorporated within the MANP framework to derive modified MANP models for problems in which the mean-field approximation fails. Numerical results on the charge dynamics with such beyond mean-field effects in inhomogeneous dielectric environments are presented to demonstrate the performance of the MANP models in the description of charge dynamics, illustrating that the proposed MANP model provides a general framework for modeling charge dynamics.


\bigskip

\noindent
{\bf Key words and phrases}:
Maxwell-Amp\`{e}re Nernst-Planck equations; Beyond mean field; Energy dissipation; Local curl-free relaxation; Linear complexity

\noindent
{\bf AMS subject classification}: \, 35K55,  82C21
\end{abstract}

\section{Introduction}

Charge dynamics plays a fundamental role in various real-world applications, e.g., electrochemical devices~\cite{BTA:PRE:04}, semiconductors~\cite{2012semiconductor}, microfluidics~\cite{Schoch:RMP:08}, and biological ion channels~\cite{E:CP:98}. At the continuum level, charge dynamics is often explored by the well-known Poisson--Nernst--Planck (PNP) theory. Based on the mean-field approximation, the classical PNP theory treats ions as point charges, takes the screening effect of background medium into account as a continuum of certain dielectric coefficients, and neglects direct ion-ion interaction details, with ions interacting only through the mean electrostatic potential. With such approximations, the Nernst-Planck equations are used to describe the diffusion and convection of ions under the gradient of the electric potential which in turn is governed by the Poisson's equation with charge sources arising from mobile ions and permanent charges in the system.

Despite its great success in various applications, the PNP theory fails in cases, e.g., concentrated ionic solutions of high valences, where the mean-field approximations break down~\cite{BKS+:ACIS:2009, LJX:SIAP:2018}. To address these issues, ionic steric effects and Coulomb correlations are often incorporated in the framework of PNP theory. For instance, ionic steric effects can be included based on the lattice gas theory, by considering the entropy of solvent~\cite{BAO:PRL:1997, KBA:PRE:2007, Li_Nonlinearity09, Li_SIMA09, LZ:BJ:2011}. Hard-sphere interactions between ions can be modeled by the Lennard-Jones potential energy, which gives rise to a nonlocal modified PNP model~\cite{HyonLiuBob_CMS10}. To get a computationally tractable model, the nonlocal Lennard-Jones integrals are further approximated to obtain a type of local modified PNP models~\cite{HorngLinLiuBob_JPCB12, LinBob_CMS14}. The ionic steric effects can also be described by the fundamental measure density functional theory~\cite{JZWu_JPCM14, Roth:JPC:2002,Wu:JCP:2002}, which gives rise to a complicated model with modified nonlocal PNP equations. To account for Coulomb correlation effects, the self energy of ions governed by the generalized Debye--H\"{u}ckel equation is taken into account in a modified PNP theory~\cite{LJX:SIAP:2018}. Such a theory further with ionic steric effects described by the fundamental measure density functional theory is able to accurately reproduce molecular simulation results on ionic distributions next to charged surfaces with dielectric mismatches~\cite{LJX:SIAP:2018, MXZ:SIAP:2021}.

It is of interest to observe that the electric potential appears in the PNP equations in the form of its gradient, i.e., the electric field. This inspires us to use ionic concentrations and the electric field, rather than the electric potential, to describe the charge dynamics. {  The replacement of the Poisson's equation by the Maxwell-Amp\`{e}re equation and further coupling with the Nernst-Planck equation to describe charge dynamics have been first proposed in a series of works~\cite{Eisenberg2020, HorngBiophys} by Eisenberg and his collaborators.}
Also, it has been shown that the formulation based on concentrations and the electric field has several advantages~\cite{BSD:PRE:2009, RM:PRL:2004, ZWL:PRE:2011}. For instance, the free-energy functional of a charged system is convex with respect to ionic concentrations~\cite{ MaggsPB_EPL12}. The convexity is a valuable feature in the design of numerical methods for finding the steady state of a charged system via minimization~\cite{ BSD:PRE:2009, MaggsPB_EPL12,PujosMaggs_JCTC15}. In this work, we propose a Maxwell-Amp\`{e}re Nernst-Planck (MANP) model in which the ionic concentrations and electric displacement, i.e., dielectric coefficient times the electric field, are the unknown variables. 
In order to overcome the drawbacks due to the mean-field approximations, ionic steric effects and Coulomb correlations are taken into account  in the framework of the MANP model. We present modified MANP models with excess energies to account for steric effects, e.g.,  described by lattice-gas theories~\cite{BAO:PRL:1997}, and Coulomb correlations modeled by dielectric self energy of ions and ion-ion correlations~\cite{LJX:SIAP:2018}.

It is also derived that the MANP model {  is energy dissipative with respect to} a convex free-energy functional. By the energy dissipation law, the steady state of the MANP model reproduces the charge conserving PB theory~\cite{Lee_JMP2014,Lee_NonL2010, Wan_PRX2014}. Therefore, our MANP model provides an alternative energy stable approach to study the charge conserving PB theory analytically and numerically.  The electric displacement obtained from the discretized Maxwell-Amp\`{e}re equation is updated with a local relaxation procedure to achieve the curl-free property. It is shown that the local curl-free relaxation algorithm is guaranteed to converge and is of linear computational complexity. Such a local relaxation idea was first proposed in~\cite{M:JCP:2002, MR:PRL:2002} to devise a local molecular simulation algorithm for the long-range Coulomb interactions. It is later coupled with the Maxwell equations to obtain the Maxwell equation molecular dynamics (MEMD) for dynamics simulations~\cite{PD:JPCM:2004,RM:JCP:2004MC, RM:PRL:2004}. The algorithm has also been extended to deal with inhomogeneous dielectrics~\cite{FH:PRE:2014,FXH:JCP:14} and solve the PB equations~\cite{BSD:PRE:2009} as well as the PB equations with steric effects~\cite{ZWL:PRE:2011}.

The advantage of the MANP formulation is that it is a local formulation and avoids solving the Poisson's equation for which the dielectric coefficient could be inhomogeneous in cases of practical interests. In each time step, the electric displacement can be explicitly updated by the discretized Maxwell-Amp\`{e}re equation with the updated ionic concentrations. Explicit formula for each local relaxation step is available in the local curl-free relaxation algorithm. Also, numerical tests demonstrate that the number of relaxation steps in the local relaxation algorithm is small and bounded. Therefore, the  update algorithm for the electric displacement in each time step is of linear computational complexity. As aforementioned, the MANP model can efficiently deal with space-dependent permittivity instead of solving the variable-coefficient Poisson's equation. In addition, due to the local nature of the algorithm for the electric displacement, the MANP model and corresponding numerical methods can be further extended to the efficient coarse-grained simulations of charge dynamics with fluctuations~\cite{TorreEspanol_JCP2011}.  It also has great potential to apply the MANP model to other areas such as the particle-in-cell simulation of plasma~\cite{chen2011JCP:VA, Degond_JCP2010}.  For example, a Vlasov--Amp\`{e}re model, rather an earlier Vlasov--Poisson model, has been applied to particle-in-cell simulations of plasma~\cite{chen2011JCP:VA}. In contrast to the one-dimension simulations performed in the work~\cite{chen2011JCP:VA}, our MANP model with the local curl-free relaxation algorithm is promising to deal with high dimensional cases.

The rest of this paper is organized as follows. In Sec.~\ref{s:MANP}, we derive the MANP model. In Sec.~\ref{s:prop}, we present some properties and advantages of the MANP model. Sec.~\ref{s:Results} presents numerical methods for the MANP model and some results.  The paper ends with some conclusions in Sec.~\ref{s:Con}.

\section{Model}\label{s:MANP}
This section begins with an overview of the PNP equations, then advances to a reformulated
model, namely, the MANP equations, and finally shows the equivalence of the two models.
In spite of the mathematical equivalence, we remark that the MANP framework has its advantages in the incorporation of
many-body effects and the design of numerical schemes.

\subsection{Poisson-Nernst-Planck equations}

Consider an electrolyte solution of $M$ ionic species occupying a bounded and connected domain $\Omega$  with periodic boundary conditions.
Let $c^\ell(\bm{r},t)$ be the ionic concentration for the $\ell$-th species ($\ell=1,\dots, M$) at time $t$, and $\phi(\bm{r},t)$ be the electric potential. The electrostatic free energy of the system is given by
\begin{equation}
	F[c^1, \dots, c^M]=\int_\Omega \left[ \frac{\ve_0 \ve_{r} |\nabla\phi|^2}{2} + k_B T \sum_{\ell=1}^{M} c^\ell \left( \log (\Lambda^3 c^\ell ) -1 \right) \right] d\bm{r},
\end{equation}
where $\ve_0$ is the vacuum permittivity, $\ve_{r}$ is the relative permittivity that can be spatially dependent, $k_B T$ is the thermal energy, and $\Lambda$ is the thermal de Broglie wavelength. The electrostatic free energy $F$ consists of the electric field energy and entropic contribution of ions. Here, the electric potential $\phi$ is determined by the ionic concentrations through the Poisson's equation
\begin{equation}\label{poisson}
-\nabla \cdot \ve_0 \ve_r \nabla \phi =\rho,
\end{equation}
where $\rho=\sum_{\ell=1}^M z^{\ell} e c^{\ell}+\rho^f$ is the total charge density, $z^{\ell}$ is the valence, and $\rho^f(\bm{r}): \Omega \to \R$ is a given function representing the distribution of fixed charges.
Ion transport in the system is described by the conservation law
\begin{equation}\label{np}
	\frac{\partial c^{\ell}}{\partial t}=-\nabla \cdot \bm{J}^{\ell},
\end{equation}
where the ionic flux $\bm{J}^{\ell}$ is given by
$$\bm{J}^{\ell}=-\frac{\gamma^{\ell} c^{\ell}}{k_B T} \nabla \mu^{\ell},$$
and $\gamma^{\ell}>0$ is the diffusion coefficient. The chemical potential is given by the variation of the free energy,
$\mu^{\ell}=\delta F/\delta c^\ell$, which is,
\begin{equation}
	\mu^{\ell}= z^{\ell}  e \phi+ k_B T\log (\Lambda^3 c^\ell ).
\end{equation}
It follows from periodic boundary conditions that the total ionic mass of each species conserves in the domain:
\[
\frac{d}{dt} \int_\Omega c^\ell d\bm{r} = 0.
\]
Therefore, we have $\int_\Omega c^\ell (\bm{r}, t) d\bm{r} = \int_\Omega c^\ell (\bm{r}, 0) d\bm{r} =N^{\ell}$, where $N^{\ell}$ is the total ionic mass of $\ell$-th ionic species.  Coupling \eqref{np} with the Poisson's equation \reff{poisson} leads to the widely known Poisson--Nernst--Planck equations
\begin{equation}
\left\{
\begin{aligned}
& \frac{\partial c^{\ell}}{\partial t}=\nabla \cdot \gamma^\ell\left(  \nabla c^\ell+\frac{z^\ell e c^\ell}{k_B T} \nabla \phi\right), ~\ell =1, \cdots, M, \\
& -\nabla \cdot \ve_0 \ve_r \nabla \phi =\rho.
\end{aligned}
\right.
\end{equation}

We introduce characteristic concentration $c_0$, diffusion constant $\gamma_0$, and lengths $L$ and $l_c=\left[ \ve_0 k_B T/(e^2 c_0)  \right]^{1/2}$.   Define $\tilde{\bm{r}} = \bm{r}/L$, $\tilde{t} = t\gamma_0 / (l_c L)$, $\tilde{\nabla}= L \nabla$, $\tilde{\rho}^f= \rho^f/(ec_0)$, $\tilde{\gamma} = l_c \gamma^l/ (L\gamma_0)$, $\ve= \ve_r l_c^2/ L^2$, $\tilde{\Lambda}^3 = \Lambda^3 c_0$ $\tilde{c}^\ell= c^\ell/c_0$, $\tilde{\phi}=e \phi /(k_B T )$, $\tilde{N}^{\ell} = N^{\ell}/(L^3 c_0)$, and $\tilde{F}=F/(k_B T c_0 L^3)$. After rescaling, we have the nondimensionalized PNP equations
\begin{equation}
	\label{PNP}
	\left\{
	\begin{aligned}
		& \frac{\partial \tilde{c}^{\ell}}{\partial \tilde{t} }=\tilde{\nabla} \cdot \tilde{\gamma}^\ell\left(\tilde{\nabla}  \tilde{c}^\ell+z^\ell  \tilde{c}^\ell \tilde{\nabla} \tilde{\phi}\right), ~\ell =1, \cdots, M, \\
		& -\tilde{\nabla} \cdot \ve \tilde{\nabla} \tilde{\phi} =\sum_\ell^M z^\ell \tilde{c}^\ell+ \tilde{\rho}^f,
	\end{aligned}\right.	
\end{equation}
and the electrostatic free energy

\begin{equation}\label{E}
		\tilde{F}[\tilde{c}^1, \dots, \tilde{c}^M]=\int_{\tilde{\Omega}} \left[ \frac{\varepsilon | \tilde{\nabla} \tilde{\phi}|^2}{2} +  \sum_{\ell=1}^{M} \tilde{c}^\ell \log \tilde{c}^\ell  \right] d\tilde{\bm{r}} + (\log\tilde{\Lambda}^3 -1) \sum_{\ell=1}^{M} \tilde{N}^{\ell},
\end{equation}
with the second equation in~\reff{PNP} as the constraint. The last term in~\reff{E} is a constant that will be ignored for simplicity. Also, we will drop all the tildes
and use the dimensionless equations for our discussion in the rest of the paper.

\subsection{Maxwell--Amp\`{e}re Nernst--Planck equations}
In the PNP equations, charge dynamics of the system is described by the ionic concentrations and electric potential. Based on an observation that the electric potential appears in the PNP equations in the form of its gradient, we can alternatively work on ionic concentrations and the electric displacement
$\bD=\ve \bE$ instead, where $\bE=-\nabla \phi$ is the electric field. The charge continuity equation reads,
\begin{equation}
	\frac{\partial \rho}{\partial t}+\nabla\cdot \bm{j} =0,
\end{equation}
where the current density of charge is,
\begin{equation}
\bm{j}=-\sum_{\ell=1}^{M} z^\ell \gamma^\ell c^\ell \nabla \mu^\ell.
\end{equation}
It follows from the Gauss's law $\nabla \cdot \bD = \rho$ that
\begin{equation}
	\nabla\cdot \left(\frac{\partial \bm{D}}{\partial t} +\bm{j} \right)=0.\label{Gs}
\end{equation}
{  Such a derivation, combining the charge continuity equation and the Poisson's equation, was first presented in~\cite{HorngBiophys}}. It then leads to the Maxwell-Amp\`{e}re equation
\begin{equation}
	\frac{\partial \bm{D}}{\partial t} +\bm{j}= \bm{\Theta},
\end{equation} 
where $\bm{\Theta}$ satisfies the Coulomb gauge condition $\nabla \cdot \bm{\Theta}= 0$. Also, it is introduced as a degree of freedom to enforce $\nabla \times (\bm{D}/\varepsilon)=\bm{0}$, which implies the existence of the electric potential satisfying the Poisson's equation in a connected spatial domain. With the electric displacement $\bm{D}$, the charge dynamics can be described by the MANP equations
\begin{equation}\label{MANP}
\left\{
\begin{aligned}
& \frac{\partial c^{\ell}}{\partial t}=\nabla \cdot \gamma^\ell\left(\nabla c^\ell-\frac{z^\ell c^\ell \bm{D}}{\varepsilon}\right),~  \ell =1, \cdots, M, \\
&\frac{\partial \bm{D}}{\partial t} =\sum_{\ell=1}^M z^\ell \gamma^\ell \left(\nabla c^\ell-\frac{z^\ell c^\ell \bm{D}}{\varepsilon}\right) +\bm{\Theta}, \\
&\nabla \cdot \bm{\Theta}=0, \\
& \nabla \times\frac{\bm{D}}{\varepsilon}=\bm{0},
\end{aligned}
\right.
\end{equation}
with certain initial concentrations $c^\ell(\bm{r},0)$ and electric displacement $\bm{D}(\bm{r},0)$.

It is not difficult to show that the PNP equations are equivalent to the MANP equations~\reff{MANP}. Taking divergence of both sides of the Maxwell-Amp\`{e}re equation, we obtain
\begin{equation}
	\nabla \cdot \left( \frac{\partial \bm{D}}{\partial t}\right)=\sum_{\ell=1}^{M} z^\ell \frac{\partial c^\ell}{\partial t}.
\end{equation}
Integrating it with respect to time results in the Gauss's law
\begin{equation}\label{ReprodGauss}
	\nabla \cdot  \bm{D}=\sum_{\ell=1}^{M} z^\ell c^\ell + a(\bm{r}),
\end{equation}
where $a(\bm{r})$ is a time-independent function.  From the curl-free equation in~\reff{MANP},  there exists a scalar function $\phi$ such that $\bm{D}/\varepsilon=-\nabla \phi$, which can be combined with the Gauss's law to recover the Poisson's equation
\[
-\nabla \cdot (\varepsilon \nabla \phi)=\sum_{\ell=1}^{M} z^\ell c^\ell  + a(\bm{r}).
\]
If the initial electric displacement is prescribed by $\bm{D}(\bm{r},0) = -\ve \nabla \phi(\bm{r},0)$ with $\phi(\bm{r},0)$ determined by
\[
-\nabla \cdot \left[\varepsilon \nabla \phi(\bm{r},0)\right]=\sum_{\ell=1}^{M} z^\ell c^\ell (\bm{r},0)  + \rho^f,
\]
then we can identify $a(\bm{r})=\rho^f(\bm{r})$ as the distribution function of permanent charges. The Nernst-Planck equations can be readily recovered with $\bm{D}=-\ve \nabla \phi$.

\section{Properties of the MANP formulation} \label{s:prop}
\subsection{Energy dissipation}
It follows from~\reff{E} that the electrostatic free energy can also be expressed as
\begin{equation}\label{EDc}
\mathcal{F}[c^1, \dots, c^M] = \int_{\Omega} \left (\frac{|\bm{D}|^2}{2\varepsilon} + \sum_{\ell=1}^{M} c^\ell\log c^\ell  \right) d\bm{r}, \mbox{ with } \nabla \cdot \bm{D}=\sum_{\ell=1}^M z^{\ell} c^{\ell}+\rho^f.
\end{equation}
Obviously, $\mathcal{F}[c^1, \dots, c^M]$ is a convex energy functional {  and the proposed MANP model can be shown to be energy dissipative}.
\begin{thm}\label{dE/dt}
The MANP equations~\reff{MANP} satisfy the energy dissipation law
\begin{equation} \label{disspt}
\frac{d\mathcal{F}[c^1, \dots, c^M]}{dt} = - \omega \leq 0,
\end{equation}
where the energy dissipation functional
\begin{equation} \label{dissptfunc}
\omega = \sum_{\ell=1}^M \int_{\Omega} \frac{\gamma^\ell c^\ell}{\ve^2} \left| \ve \nabla \log c^\ell - z^\ell \bm{D}  \right|^2 d\bm{r}.
\end{equation}
\end{thm}
\begin{proof}
Taking derivative of $\mathcal{F}$ with respect to time, we have
\[
\frac{d\mathcal{F}}{dt} = \int_{\Omega} \left[ \frac{\bm{D}}{\varepsilon} \cdot \frac{\partial \bm{D}}{\partial t} + \sum_{\ell=1}^{M} \frac{\partial c^\ell}{\partial t} \left( \log c^\ell+1 \right) \right] d\bm{r}.
\]
It follows from the MANP equations~\reff{MANP} that
\[
\begin{aligned}
\frac{d\mathcal{F}}{dt} &=  \int_{\Omega} \left[ \frac{\bm{D}}{\varepsilon} \cdot \left( \bm{\Theta}- \bm{j} \right) + \sum_{\ell=1}^{M} \nabla \cdot \gamma^\ell\left(\nabla c^\ell-\frac{z^\ell c^\ell \bm{D}}{\varepsilon}\right)  \left( \log c^\ell+1 \right)  \right] d\bm{r} \\
                    &= \int_{\Omega} \left[ \sum_{\ell=1}^{M}  z^\ell \gamma^\ell \frac{\bm{D}}{\varepsilon} \cdot \left(\nabla c^\ell-\frac{z^\ell c^\ell \bm{D}}{\varepsilon}\right) - \sum_{\ell=1}^{M} \frac{\gamma^\ell}{c^\ell} \left(\nabla c^\ell-\frac{z^\ell c^\ell \bm{D}}{\varepsilon}\right) \cdot \nabla c^\ell \right] d\bm{r}\\
                    &= -\sum_{\ell=1}^M \int_{\Omega} \frac{\gamma^\ell }{\ve^2 c^\ell} \left| \ve \nabla  c^\ell - z^\ell c^\ell \bm{D}  \right|^2 d\bm{r} \\
                    &= - \omega \leq 0,
\end{aligned}
\]
where we have used the integration by parts, periodic boundary conditions, and the equality $\displaystyle\int_{\Omega} (\bm{D}/\varepsilon) \cdot \bm{\Theta} d\bm{r} = 0$. The latter equality can be shown by the facts
that $\nabla \times (\bm{D}/\varepsilon)=\bm{0}$, $\nabla \cdot \bm{\Theta}=0$, and periodic boundary conditions.
\end{proof}
\begin{rem}
When non-periodic boundary conditions are considered, the energy evolution of the MANP formulation becomes
\[
\frac{d\mathcal{F}}{dt} =  - \omega + \int_{\partial \Omega} \left( \phi \frac{\partial \sigma}{\partial t} +\sum_{\ell=1}^{M} \gamma^\ell c^{\ell}  \mu^{\ell} \frac{\partial \mu^{\ell}}{\partial \bm{n}} \right) d S,
\]
where $\bm{n}$ denotes the unit exterior normal at the boundary and $\sigma =-\bm{D} \cdot \bm{n}$ represents the surface charge density at the boundary. Again, one can obtain energy dissipation $d\mathcal{F}/dt \leq 0$, if zero normal ionic flux boundary conditions are imposed for a closed system and time-independent surface charge density is prescribed as a boundary condition for the electric displacement.
\end{rem}

With such an {  energy dissipative law}, it is desirable to design numerical methods that are able to preserve energy dissipation at discrete level. Structure-preserving numerical methods for the MANP model will be studied in future works.

\subsection{Steady state}
The convex free energy $\mathcal{F}[c^1, \dots, c^M]$ can be shown to be bounded below by using the inequality $x \log x \geq - 1/e$. Therefore, the existence of a steady state of the MANP equations can be established by  the energy dissipation law; cf.~Theorem~\ref{dE/dt}. It follows from~\reff{disspt}, \reff{dissptfunc}, \reff{MANP}, and~\reff{ReprodGauss} that the steady state is characterized by equations
\begin{equation}\label{PBSyst}
\left\{
\begin{aligned}
&\ve \nabla c^\ell -z^\ell c^\ell  \bm{D}=0,\\
&\nabla \cdot \bm{D}=\sum_{\ell=1}^M z^\ell c^\ell +\rho^f,\\
&\nabla \times \frac{\bm{D}}{\ve}=0.
\end{aligned}\right.
\end{equation}
Such a system is equivalent to the well-known Poisson--Boltzmann equation. Indeed, from the curl-free condition we have that the electrostatic potential satisfies $\bm{D}/\varepsilon=-\nabla \phi$.  By the first equation in~\reff{PBSyst}, we obtain the Boltzmann distributions
\[
c^\ell = c_N^{\ell} e^{-z^\ell \phi},
\]
where $c_N^{\ell} = N^{\ell}/\int_{\Omega} e^{-z^\ell \phi} d\bm{r}$ is a normalization constant determined by the total ionic mass $N^{\ell}$. Substituting the expressions of $\bm{D}$ and $c^\ell$ into the second equation of system~\reff{PBSyst}, we obtain the charge conserving PB (ccPB) equation~\cite{Lee_JMP2014,Lee_NonL2010, Wan_PRX2014}
\begin{equation}
-\nabla \cdot (\varepsilon \nabla \phi)=\sum_{\ell=1}^{M} \frac{z^\ell N^{\ell}}{\int_{\Omega} e^{-z^l \phi} d\bm{r}}  e^{-z^l \phi}  + \rho ^f.
\end{equation}
Note that, if an open system is connected to a bulk ionic solution reservoir, the normalization constant can be determined with the bulk ionic concentrations $c^{\ell, \infty}$, leading to the classical Boltzmann distributions $c^\ell = c^{\ell, \infty} e^{-z^\ell \phi}.$ Then, we can obtain the classical PB equation.

As the governing equation for the steady state,  the ccPB equation can be solved from a variational perspective. The convex structure plays a crucial role in the design of efficient numerical methods for finding steady states via minimization~\cite{BSD:PRE:2009, PujosMaggs_JCTC15}. In literature, the energy functional of the electrostatic potential is known to be non-convex~\cite{CDLM_JPCB08}. In order to have a minimizing principle, rather a stationary principle, the Legendre transform was used to formulate a local, convex dual functional~~\cite{MaggsPB_EPL12}. Our MANP formulation provides a promising approach to investigate the steady state. Energy dissipating numerical schemes that allow large time step sizes can be used to find the steady state robustly and efficiently, especially for the system with conserved total ionic mass.

\subsection{Beyond mean field}
In mean-field approximations, ions are treated as point charges without direct ion-ion interactions, e.g., steric interactions of short range and Coulombic correlations of long range. Therefore, the MANP theory proposed above does not work well for concentrated ionic solutions of high valences, where the mean-field approximations break down.  To overcome the defects, ionic steric effects and Coulomb correlations can be included within our MANP framework. Consider a free-energy functional
\begin{equation}\label{F-BMF}
\mathcal{F}[c^1, \dots, c^M] = \int_{\Omega} \left (\frac{|\bm{D}|^2}{2\varepsilon} + \sum_{\ell=1}^{M} c^\ell\log c^\ell  \right) d\bm{r} + \mathcal{F}^{st}  +  \mathcal{F}^{co},
\end{equation}
where $\mathcal{F}^{st}$ describes the excess free energy due to steric effects of finite-size ions and $\mathcal{F}^{co}$ represents the excess Coulomb interaction energy.  Such two excess energies lead to the corresponding excess chemical potentials defined by
\[
\mu^{\ell, {st}}=\frac{\delta F^{st}}{\delta c^\ell}~~ \text{and } ~ \mu^{\ell, {co}}=\frac{\delta F^{co}}{\delta c^\ell}.
\]

There are several models to describe the excess energy due to steric effects in literature. One popular model is based on the statistical mechanics of ions and solvent molecules on a cubic lattice of spacing in the continuum limit~\cite{BKS+:ACIS:2009}.  In such a model,  the entropy of solvent molecules, in addition to the entropy of ions, accounts for the excess chemical potential~\cite{Bikerman:PM:1942, BAO:PRL:1997, LLXS:NonL:2013, LZ:BJ:2011}
\[
\mu^{\ell, {st}}= - \frac{v_{\ell}}{v_0} \ln \left(v_0 c_0 \right),
\]
where $v_\ell$ and $v_0$ are the volume of ions of the $\ell$-th species and solvent molecule, respectively,  and $c_0 = c_0(\mathbf{r},t)$ is the solvent concentration defined by
\[
c_0(\mathbf{r},t) = v_0^{-1} \left[ 1 - \sum_{\ell=1}^M v_\ell c^{\ell}(\mathbf{r},t) \right].
\]
Alternatively, the steric effects due to hard-sphere interactions can be described by the Lennard-Jones potential energy, which gives rise to a nonlocal model~\cite{HyonLiuBob_CMS10}. Local approximations of nonlocal Lennard-Jones integrals can be employed to obtain local models by Fourier analysis~\cite{HorngLinLiuBob_JPCB12, LinBob_CMS14}. The excess chemical potential is then given by
\begin{equation}\label{LJmu}
\mu^{\ell, {st}}= \sum_{j=1}^M \omega_{\ell j}c_j,
\end{equation}
where $(\omega_{ij})_{M\times M}$ is a symmetric matrix to represent the cross interactions between different species and self interactions between ions of the same species~\cite{DingWangZhou_JCP2019}.  Also, there are other complicated models to describe steric effects, such as the fundamental measure density functional theory~\cite{JZWu_JPCM14,Roth:JPC:2002,Wu:JCP:2002}.

In addition to ion-ion correlation, the Coulomb correlation includes the dielectric self energy of ions.  The excess chemical potential due to Coulomb correlation is given by~\cite{LJX:SIAP:2018,MX:JCP:14,MXZ:SIAP:2021,XML:PRE:2014}
\begin{equation}\label{muco}
\mu^{\ell, {co}}= \frac{(z^\ell )^2}{2} \left[ \frac{1}{4\pi a^\ell} \left(\frac{1}{\ve (\bm{r}) } - 1 \right) +  \lim _{\bm{r}^{\prime} \rightarrow \bm{r}}\left (G\left(\bm{r}, \bm{r}^{\prime}\right)- \frac{1}{4\pi \ve (\bm{r})  |\bm{r} - \bm{r}^{\prime} |} \right) \right],
\end{equation} 
where $a^\ell$ is the radius for ions of the $\ell$-th species, and $G(\bm{r}, \bm{r}^{\prime})$ is the Green's function satisfying the generalized Debye-H\"{u}ckel equation
\begin{equation} \label{gdh}
	- \nabla \cdot \ve (\bm{r})   \nabla G\left(\bm{r}, \bm{r}^{\prime}\right)+2 I(\bm{r},t)  G\left(\bm{r}, \bm{r}^{\prime}\right)=\delta\left(\bm{r}-\bm{r}^{\prime}\right).
\end{equation}
To consider boundary effects, we consider inhomogeneous dielectric permittivity and the ionic strength profile given by
\[
	I(\bm{r},t)=\left\{
	\begin{aligned}
		&\frac{1}{2} \sum_{\ell=1}^M (z^\ell)^{2} c^\ell (\bm{r},t),~ &\bm{r} \in \Omega, \\
		&0,  &\bm{r} \notin \Omega.
	\end{aligned}\right.
\]

With the excess contributions from steric effects and Coulomb correlation, our MANP model becomes
\begin{equation}\label{MANP-Beyond}
\left\{
\begin{aligned}
& \frac{\partial c^{\ell}}{\partial t}=\nabla \cdot \gamma^\ell\left(\nabla c^\ell-\frac{z^\ell c^\ell \bm{D} }{\varepsilon} + c^\ell \nabla \mu^{\ell, {st}}+ c^\ell \nabla \mu^{\ell, {co}} \right),~  \ell =1, \cdots, M, \\
&\frac{\partial \bm{D}}{\partial t} =\sum_{\ell=1}^M z^\ell \gamma^\ell \left(\nabla c^\ell-\frac{z^\ell c^\ell \bm{D} }{\varepsilon} + c^\ell \nabla \mu^{\ell, {st}}+ c^\ell \nabla \mu^{\ell, {co}} \right)+\bm{\Theta}, \\
&\nabla \cdot \bm{\Theta}=0, \\
& \nabla \times\frac{\bm{D}}{\varepsilon}=\bm{0}.
\end{aligned}
\right.
\end{equation}
	
\begin{rem} Solving the generalized Debye-H\"uckel equation \eqref{gdh} is challenging due to the high-dimensional Green's function $G(\bm{r},\bm{r}^{\prime})$.
For channel problems which are often studied in nanodevices, asymptotic solutions by using Wentzel-Kramers-Brillouin approximation
can be employed; See \cite{MXZ:PRE:2021,MXZ:SIAP:2021} for the recent work. Moreover, the Coulomb correlation \eqref{muco} only requires the
self Green's function. Numerically, this corresponds to the diagonal of the inverse of the operator matrix and
can be calculated efficiently with linear scaling \cite{Tu:JCP:2022} by coupling the selective inversion with the hierarchical interpolative factorization.
\end{rem}

\subsection{Local curl-free relaxation}\label{localsec}
One advantage of using the MANP model is featured by a local relaxation algorithm that is Gauss-law satisfying, which can solve the curl-free constraint in the MANP equations  efficiently and robustly.
The local curl-free relaxation algorithm is derived by considering the minimization of a convex electric field energy
\begin{equation}\label{MinPro}
\mathcal{F}_{\rm pot}[\bm{D}] :=  \int_{\Omega} \frac{|\bm{D}|^2}{2\varepsilon} d\bm{r}
\end{equation}
with $\bm{D}$ subject to the Gauss' law.
Define the Lagrangian
\[
\mathcal{L}[\bm{D}, \lambda]:= \int_{\Omega} \frac{|\bm{D}|^2}{2\varepsilon} d\bm{r} - \int_{\Omega} \lambda  \left(\nabla \cdot \bm{D}-\rho \right) d\bm{r},
\]
where $\lambda$ is the Lagrange multiplier. By the Lagrange multiplier method, the minimizer of the constraint optimization problem satisfies
\[
\bm{D}= -\ve \nabla \lambda~\mbox{and}~ \nabla \cdot \bm{D}=\rho.
\]
This implies that the Lagrange multiplier $\lambda$ is the electrostatic potential and $\nabla \times\bm{D}/\varepsilon=\bm{0}$. Therefore, in order to achieve the curl-free condition, one can design an algorithm that can decrease the electric field energy most in each update step for $\bm{D}$.
Another perspective is based on the general solution to the Gauss' law: $\bm{D}/\varepsilon=-\nabla \phi + \nabla \times \bm{Q}$, where $\bm{Q}$ is the rotational degrees of freedom. As such,
$$
\mathcal{F}_{\rm pot}[\bm{D}] = \int_{\Omega} \frac{\varepsilon}{2} \left[\left(\nabla \phi \right)^2+ \left(\nabla \times \bm{Q}\right)^2 \right] d\bm{r}.
$$
Thus, the minimization corresponds to the elimination of the rotational degrees of freedom in $\bm{D}$, which is dominated by a torque~\cite{MR:PRL:2002}.

Specifically, we consider spatial discretization in a rectangular domain $\Omega=[0, L_1]\times [0,L_2]$ with periodic boundary conditions. The domain is covered by a uniform grid  with the grid spacing $\Delta x$ and $\Delta y$ in the two coordinate directions. Denote by $\Delta \Omega =\Delta x \Delta y$ the area of a unit cell. Let $N_x$ and $N_y$ be the number of grid points along each dimension. Let $\bm{D}=(D_x, D_y)$. In particular, $c_{i,j}^{\ell}$ stands for the numerical approximation of $c^{\ell}$ on the grid point $\left(i\Delta x, j\Delta y\right)$, and $D_{i+1/2, j}$ stands for the numerical approximation of $D_x$ on the grid point $\left((i+1/2)\Delta x , j \Delta y\right)$  for  $i=1, \ldots, N_{x}$ and  $j=1, \ldots, N_{y}$. Analogously, $D_{i, j+1/2}$ stands for the numerical approximation of $D_y$ on the grid point $\left( i\Delta x , (j+1/2) \Delta y\right)$.


\begin{figure}[h!]
	\begin{center}
		\includegraphics[scale=0.5]{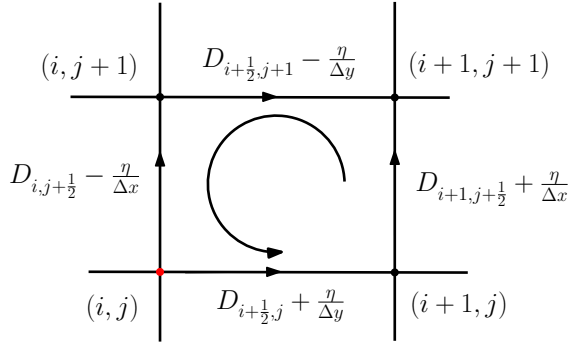}
		\caption{Update diagram of the electric displacements in a single cell.}
		\label{fD}
	\end{center}
\end{figure}

With central differencing, the electric field energy in the optimization problem~\reff{MinPro} is approximated with second-order accuracy by
\begin{equation}
	\mathcal{F}_{\rm pot}^{h}=\frac{\Delta {\Omega}}{2} \sum_ {i,j} \left( \frac{D^2_{i+\frac{1}{2},j}}{\varepsilon_{i+\frac{1}{2},j}}  + \frac{D^2_{i,j+\frac{1}{2}}}{\varepsilon_{i,j+\frac{1}{2}}} \right),
\end{equation}
and the Gauss's law is approximated by
\begin{equation}
	(\nabla \cdot) _h \bm{D}_{i,j}:=
	\frac{D_{i+\frac{1}{2},j}-D_{i-\frac{1}{2},j}}{\Delta x}
	+\frac{D_{i,j+\frac{1}{2}}-D_{i,j-\frac{1}{2}}}{\Delta y}
	=\sum_\ell z^\ell c^\ell_{i,j} +\rho^f _{i,j}.\label{gausslaw}
\end{equation}
Our relaxation starts from the solution to the discrete Maxwell-Amp\`{e}re equation.
Successive local updates of the electric displacements in each cell are proposed to minimize $\mathcal{F}_{\rm pot}^{h}$~\cite{MR:PRL:2002}; See a schematic diagram shown in Figure~\ref{fD}. For instance, consider a grid cell which has four nodes with indices $(i,j)$, $(i+1,j)$, $(i, j+1)$, $(i+1, j+1)$, and four respective edges. Denote by $D_{i+1/2, j}$, $D_{i+1, j+1/2}$, $D_{i+1/2, j+1}$, and $D_{i, j+1/2}$ the four electric displacement components defined on the four edges. The arrows on the edges in Figure~\ref{fD} indicate the predefined positive directions.  In the update, a flux $\eta$ that rotates along four edges is introduced in such a way that the discrete energy $\mathcal{F}_{\rm pot}^{h}$ is locally minimized. Specifically, the updates of displacements are given by
\begin{equation}
	\begin{aligned}
		D_{i+\frac12, j} &\leftarrow D_{i+\frac12, j}+\frac{\eta}{\Delta y} ,\\
		D_{i+1, j+\frac12} &\leftarrow D_{i+1, j+\frac12}+\frac{\eta}{\Delta x},\\
		D_{i+\frac12, j+1} &\leftarrow D_{i+\frac12, j+1}-\frac{\eta}{\Delta y} ,\\
		D_{i, j+\frac12}     &\leftarrow D_{i, j+\frac12}-\frac{\eta}{\Delta x}.
	\end{aligned}\label{updateD}
\end{equation}
It is remarkable that the Gauss's law is still rigorously maintained at all four nodes with such an update scheme, i.e., the constraint of the discrete Gauss's law~\reff{gausslaw} is satisfied for each update.  The associated change in the energy after updates in displacements reads
\begin{equation}\label{deltaE}
	\begin{aligned}
		\delta \mathcal{F}_{\rm pot}^{h}[\eta] = &\frac{\eta ^2}{2}
		\left[\frac{\Delta x}{\Delta y} \left(\frac{1}{\varepsilon_{i+\frac{1}{2},j}}+\frac{1}{\varepsilon_{i+\frac{1}{2},j+1}} \right) +\frac{\Delta y}{\Delta x} \left(\frac{1}{\varepsilon_{i,j+\frac{1}{2}}}+\frac{1}{\varepsilon_{i+1,j+\frac{1}{2}}} \right) \right] \\
		&+ \eta \left[\Delta x\left(\frac{D_{i+\frac{1}{2},j}}{\varepsilon_{i+\frac{1}{2},j}}-\frac{D_{i+\frac{1}{2},j+1}}{\varepsilon_{i+\frac{1}{2},j+1}}\right)+  \Delta y\left(\frac{D_{i+1,j+\frac{1}{2}}}{\varepsilon_{i+1,j+\frac{1}{2}}}-\frac{D_{i,j+\frac{1}{2}}}{\varepsilon_{i,j+\frac{1}{2}}}\right) \right],
	\end{aligned}	
\end{equation}
which is minimized with
\begin{equation}\label{eta}
	\eta= - \frac{\Delta y(\Delta x)^2\left(\frac{D_{i+\frac{1}{2},j}}{\varepsilon_{i+\frac{1}{2},j}}-\frac{D_{i+\frac{1}{2},j+1}}{\varepsilon_{i+\frac{1}{2},j+1}}\right)+  \Delta x(\Delta y)^2\left(\frac{D_{i+1,j+\frac{1}{2}}}{\varepsilon_{i+1,j+\frac{1}{2}}}-\frac{D_{i,j+\frac{1}{2}}}{\varepsilon_{i,j+\frac{1}{2}}}\right)}{(\Delta x)^2 \left(\frac{1}{\varepsilon_{i+\frac{1}{2},j}}+\frac{1}{\varepsilon_{i+\frac{1}{2},j+1}} \right) +(\Delta y)^2 \left(\frac{1}{\varepsilon_{i,j+\frac{1}{2}}}+\frac{1}{\varepsilon_{i+1,j+\frac{1}{2}}} \right)}.
\end{equation}
This gives an optimal flux for the displacement changes in terms of minimizing the discrete energy.
Such a local update loops over all the grid points for certain steps until a stopping criterion $\ve_{\rm tol}$ is met. Since analytical expressions are available for each update step, it is clear that the complexity of the algorithm is $\mathcal{O}(N)$, where $N=N_x\cdot N_y$.

This local curl-free relaxation uses the intrinsic dynamics of the electric field, i.e., the propogation of a retarded and diffusive vector field. The equilibrium of the local curl-free relaxation gives the electrostatic interactions essentially~\cite{M:JCP:2002}.   The whole relaxation can be regarded as a projection of the electric displacement from the non-curl-free space to the curl-free space, while maintaining in the Gauss' law satisfying surface at the same time.

\begin{rem}
	\begin{compactenum}
		\item[\rm (1)]
		The convergence of the local relaxation algorithm with respect to iteration steps can be proved by the fact that $\mathcal{F}_{\rm pot}^{h}$ is nonnegative and $\delta \mathcal{F}_{\rm pot}^{h}(\eta)\leq 0 $ when $\eta$ is given by~\reff{eta}.
		\item[\rm (2)]
		Taking the derivative of the updated electric field energy with respect to flux $\eta$ gives the torque which drives the circulatory dynamics. It explains why the curl-free condition is satisfied when the relaxation converges. In 2D, the discrete curl of $\bm{D}/\varepsilon$ is approximated by
		\begin{equation*}
			(\nabla \times)_h \Big( \frac{\bm{D}}{\varepsilon} \Big)_{i+\frac{1}{2},j+\frac{1}{2}}=\frac{1}{\Delta x} \bigg(\frac{D_{i+1,j+\frac{1}{2}}}{\varepsilon_{i+1,j+\frac{1}{2}}}-\frac{D_{i,j+\frac{1}{2}}}{\varepsilon_{i,j+\frac{1}{2}}}\bigg) + \frac{1}{\Delta y}  \bigg(\frac{D_{i+\frac{1}{2},j}}{\varepsilon_{i+\frac{1}{2},j}}-\frac{D_{i+\frac{1}{2},j+1}}{\varepsilon_{i+\frac{1}{2},j+1}}\bigg).
		\end{equation*}
		It follows from~\reff{deltaE} and~\reff{eta} that $\delta \mathcal{F}_{\rm pot}^{h}=0$  is achieved if and only if $\eta =0$, which is in turn equivalent to $(\nabla \times)_h \left( \bm{D}/\varepsilon\right)_{i+1/2,j+1/2}=0.$ Therefore, the curl-free constraint is achieved as the local relaxation algorithm converges.
		
		\item[\rm (3)] Starting from $\bm{D}^{*}$, the local curl-free relaxation gives rise to the updated $\bm{D}^{n+1}\in \bm{D}^{*}+ \mathbb{V}$, where the space $\mathbb{V}:= \{ \bm{v} | (\nabla \cdot)_h \bm{v} =0 \}$.
		
		\item[\rm (4)]
		It is easy to extend this relaxation algorithm to nonuniform grids by considering the integration of the Gauss's law.
		
		\item[\rm (5)]
	    In each update step, an explicit expression for $\bm{D}$ is available. Numerical tests in Sec.~\ref{accuracy} demonstrate that with a good guess for $\bm{\Theta}$ using an extrapolation, the local curl-free relaxation only takes a few iteration steps. Thus, the local curl-free relaxation algorithm is expected to have linear complexity.	
	\end{compactenum}
	
\end{rem}

\section{Results}\label{s:Results}
In this section, we shall present the whole numerical method for the MANP model and simulation results to demonstrate the attractive features of the model.
\subsection{Numerical method}
Let $\Delta t$ be the time step size. Denote by $c^{\ell,n}$ and $\bm{D}^{n}$ the approximations to $c^\ell$ and $\bm{D}$, respectively, at time $t_n:=n\Delta t$ for a nonnegative integer $n$. An Euler discretization of the time derivative leads to a semi-implicit scheme for the Nernst--Planck equations,
$
c^{\ell,n+1}=c^{\ell,n}-\Delta t\nabla \cdot \bm{J}^{\ell, n}, \label{disnp}
$
with the flux
$
\bm{J}^{\ell, n} = -\gamma^\ell\left(\nabla c^{\ell,n+1}- z^\ell c^{\ell,n+1} \bm{D}^n/\varepsilon\right).
$
Similarly, the Maxwell-Amp\`{e}re equation is discretized by
\begin{equation}\label{disam}
\bm{D}^{*}=\bm{D}^{n} + \Delta t\left(-\sum _{\ell=1}^M z^\ell \bm{J}^{\ell, n} +\bm{\Theta}^n\right),
\end{equation}
where $\bm{D}^{*}$ is an intermediate approximation to $\bm{D}^{n+1}$. Then, $\bm{D}^{*}$ is further corrected to satisfy the curl-free constraint in~\reff{MANP}. See Section~\ref{localsec} for the details. Notice that the numerical scheme~\reff{disam} for the Maxwell-Amp\`{e}re equation is explicit once $c^{\ell,n+1}$ is obtained. There are several options for $\bm{\Theta}^n$, as long as it satisfies the divergence free condition. In our numerical simulations, we test the performance of the scheme by comparing results with the following choices
\begin{align}
	\bm{\Theta}^n_1 &=0, \label{Theta1}\\
	\bm{\Theta}^n_2 &=\frac{\bm{D}^{n} - \bm{D}^{n-1}}{{\Delta t}}+\sum _{\ell=1}^M z^\ell \bm{J}^{\ell, n-1}. \label{Theta2}
	\\ \bm{\Theta}^n_3 &=\frac{3}{2}\Big(\frac{\bm{D}^{n}-\bm{D}^{n-1}}{{\Delta t}}+\sum _{\ell=1}^M z^\ell \bm{J}^{\ell, n-1}\Big)-\frac{1}{2}\Big(\frac{\bm{D}^{n-1}-\bm{D}^{n-2}}{{\Delta t}}+\sum _{\ell=1}^M z^\ell \bm{J}^{\ell, n-2}\Big) \label{Theta3}.
\end{align}
Notice that an appropriate choice of $\bm{\Theta}^n$ can provide a better approximation $\bm{D}^{*}$ to $\bm{D}^{n+1}$, and reduce the relaxation steps needed in the correction from $\bm{D}^{*}$ to $\bm{D}^{n+1}$. The effect of different choices of $\bm{\Theta}^{n+1}$ on relaxation steps will be thoroughly studied in next section.
It follows from~\reff{disam} that
\[
	 \nabla \cdot \bm{D}^*-\nabla \cdot \bm{D}^n =\sum_\ell z^\ell  (c^{\ell,n+1}-c^{\ell,n}).
\]
Thus, it is straightforward to verify that $\bm{D}^{*}$ given by~\reff{disam} satisfies the Gauss's law, i.e.,
$\nabla \cdot \bm{D}^{*}=\sum  z^\ell c^{\ell,n+1}+\rho^f,$
provided that $\bm{D}^{n}$ satisfies the Gauss's law. 

Central differencing and upwinding schemes are employed to discretize the spatial derivatives.
Our numerical method equipped with the curl-free relaxation for the MANP model is summarized in the following algorithm.
\begin{algorithm}[H] \label{alg}
	\caption{Numerical method for the MANP equations}
	{\bf Input:}  Final time $T$, stopping criterion $\ve_{\rm tol}$ for the local curl-free relaxation, initial ionic distributions $c^{\ell, 0}$, and the corresponding displacement field $\bm{D}^0$ that satisfies the Gauss's law;
	\begin{algorithmic}[1]
 		\State Solve the Nernst--Planck equations to get $c^{\ell,n+1}$, $\ell=1, \dots, M$; 
		 \State Given a divergence free $\bm{\Theta}^n$ shown in \reff{Theta1}-\reff{Theta3}, solve the Maxwell-Amp\`{e}re equation with the scheme~\reff{disam} to obtain an intermediate displacement field $\bm{D}^*$; 
 		\State Perform the local curl-free relaxation on $\bm{D}^*$ to obtain $\bm{D}^{n+1}$.
	\State If $T\leq t^{n+1}$, stops; Otherwise, go to Step 1.
	\end{algorithmic}
\end{algorithm}

\subsection{Effect of $\bm{\Theta}$}\label{accuracy}
We study the effect of different choices of $\bm{\Theta}$ on the local relaxation algorithm, and conduct a series of simulations to demonstrate the performance of our MANP model in the description of charge dynamics.

In simulations, we consider binary electrolytes in a rescaled domain $\Omega=(-1, 1)\times (-1, 1)$ with periodic boundary conditions.  The charge dynamics is described by the MANP equations
\[
\left\{
\begin{aligned}
& \frac{\partial c^{\ell}}{\partial t}=\nabla \cdot \gamma^\ell\left(\nabla c^\ell-\frac{z^\ell c^\ell \bm{D}}{\varepsilon}\right),~  \ell =1, 2, \\
&\frac{\partial \bm{D}}{\partial t} =\sum_{\ell=1}^2 z^\ell \gamma^\ell \left(\nabla c^\ell-\frac{z^\ell c^\ell \bm{D}}{\varepsilon}\right) +\bm{\Theta}, \\
&\nabla \cdot \bm{\Theta}=0, \\
& \nabla \times\frac{\bm{D}}{\varepsilon}=\bm{0}.
\end{aligned}
\right.
\]
We take the parameters $\gamma^\ell=0.01$, $z^1=1$, and $z^2=-1$. The permanent charge in the system is given by the distribution function
\[
	\begin{aligned}
		\rho^{f}(x,y)=&e^{-100\left[\left(x+\frac{1}{2}\right)^{2}+\left(y+\frac{1}{2}\right)^{2}\right]}-e^{-100\left[\left(x+\frac{1}{2}\right)^{2}+\left(y-\frac{1}{2}\right)^{2}\right]}\\
		&\quad -e^{-100\left[\left(x-\frac{1}{2}\right)^{2}+\left(y+\frac{1}{2}\right)^{2}\right]}+e^{-100\left[\left(x-\frac{1}{2}\right)^{2}+\left(y-\frac{1}{2}\right)^{2}\right]}.
	\end{aligned}
\]
We consider a rescaled dielectric coefficient given by
\[
\varepsilon(x, y)=
	\frac{\varepsilon_w-\varepsilon_m}{2}\left[\tanh \left(100 d-25\right)+1\right]+\varepsilon_m,
\]
where $\varepsilon_w$ and $\varepsilon_m$ are the dielectric constants for solvent and solute, respectively, and
\[
d=\left\{
\begin{aligned}
	& \sqrt{\left(x-\frac{1}{2}\right)^{2}+\left(y-\frac{1}{2}\right)^{2}}, ~
	0<x\leq1,  0<y\leq1, \\
	& \sqrt{\left(x-\frac{1}{2}\right)^{2}+\left(y+\frac{1}{2}\right)^{2}}, ~
	0<x\leq1, -1\leq y\leq0, \\
	& \sqrt{\left(x+\frac{1}{2}\right)^{2}+\left(y-\frac{1}{2}\right)^{2}}, ~
	-1\leq x\leq0, 0<y\leq1, \\
	&\sqrt{\left(x+\frac{1}{2}\right)^{2}+\left(y+\frac{1}{2}\right)^{2}}, ~
	-1\leq x\leq0, -1 \leq y\leq0.
\end{aligned}
\right.\]
The initial concentrations are given by
$$
c^\ell(x, y, 0)=0.1,~~\ell=1,2,
$$
and the initial electric displacement is given by $\bm{D} (x, y, 0) = - \ve \nabla \phi$, where $\phi$ is obtained by solving the Poisson's equation
\[
-\nabla \cdot (\varepsilon \nabla \phi)=\sum_{\ell=1}^{2} z^\ell c^\ell (x, y, 0) + \rho ^f(x,y).
\]

To understand the effect of different choices of $\bm{\Theta}$ on the convergence of our local curl-free relaxation algorithm, we study the number of relaxation steps, denoted by $R_s$, that is needed to meet the stopping criterion $\ve_{\rm tol}$ in each time step. Numerical tests are performed with different choices of $\bm{\Theta}$, i.e., $\bm{\Theta}^n_1$, $\bm{\Theta}^n_2$, and $\bm{\Theta}^n_3$ in~\reff{Theta1}-\reff{Theta3}.

\begin{figure}[ht]
	\centering
	\includegraphics[scale=0.55]{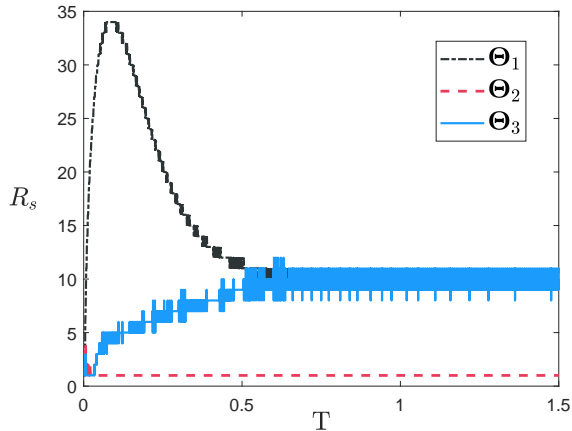}
	\caption{
		The number of relaxation steps $R_s$ against time in the local curl-free relaxation algorithm for different choices of $\bm{\Theta}$ with a stopping criterion $\ve_{\rm tol}= 1$E$-7$.
	 }
	\label{figure-theta}
\end{figure}
From Fig.~\ref{figure-theta}, one can see that, with a stopping criterion $\ve_{\rm tol}= 1$E$-7$, the number of relaxation steps $R_s$ behaves significantly differently for different $\bm{\Theta}$. For the trivial choice $\bm{\Theta}^n_1$,  $R_s$ grows rapidly in the first stage and levels off and oscillates in the late stage about $10$. For $\bm{\Theta}^n_3$, it requires less steps and $R_s$ grows to the same plateau value as that for $\bm{\Theta}^n_1$ in the late stage. It is of interest to find that $R_s$ for $\bm{\Theta}^n_2$ only takes less than $5$ steps in first a few steps and quickly decreases to $1$ step. This can be explained by the fact that the first-order Euler discretization is used in the numerical method and $\bm{\Theta}^n_2$ is a first-order extrapolation of $\bm{\Theta}$.
Overall, the study on the number of relaxation steps $R_s$ reveals that, a better estimate of the $\bm{\Theta}$ leads to $\bm{D}^*$ that is closer to divergence free, lowering the relaxation steps in the local relaxation algorithm. With $\bm{\Theta}^n_2$, the relaxation steps can be as low as $1$, indicating that the local relaxation algorithm is of linear computational complexity. It is expected that the local relaxation algorithm will be more efficient than the standard Poisson solvers for variable coefficients.


\subsection{Charge Dynamics}
In this section, we perform a series of numerical experiments to study the effects of inhomogeneous dielectrics, ionic sizes, and correlations on charge dynamics. We use a uniform mesh with $N_x=N_y=200$ and a time step $\Delta t = 0.1\Delta x$. To understand various effects, we conduct systematic numerical tests on the following cases:
\begin{itemize}
	\item Case 1: $\varepsilon_m=\varepsilon_w=2$E$-4$;
	\item Case 2: $\varepsilon_m=\varepsilon_w=2$E$-4$ with steric effects;
	\item Case 3: $\varepsilon_m=2$E$-4$ and  $\varepsilon_w=1.56$E$-2$ with Coulomb correlations.
	\item Case 4: $\varepsilon_m=2$E$-4$ and  $\varepsilon_w=1.56$E$-2$ with both steric effects and Coulomb correlations.
\end{itemize}
Solving the generalized Debye-H\"uckel equation \eqref{gdh} numerically is rather challenging in high dimensions. We here ignore the second term in~\reff{muco} for simplicity when considering Coulomb correlations.
\begin{figure}[htbp]
	\centering
	\includegraphics[scale=0.48]{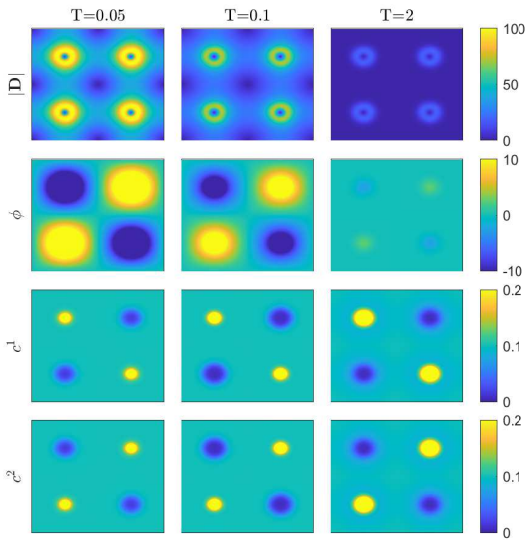}
	\caption{The snapshots of $|\bm{D}|$, $\phi$, $c^1$, and $c^2$ at time $T=0.05$, $T=0.1$, and $T=2$ for Case 1. }
	\label{f:vem=vew}
\end{figure}

In Case 1, we consider a spatially uniform dielectric coefficient with $\varepsilon_m= \varepsilon_w =2$E$-4$. Figure~\ref{f:vem=vew} shows the snapshots of the electric displacement $|\bm{D}|$, potential $\phi$, and concentrations at time $T=0.05$, $T=0.1$, and $T=2$. One can see that the concentration of ions accumulates at the permanent charges of opposite signs, due to the electrostatic attractions. Meanwhile, the magnitude of the electric displacement first develops ring peaks close to the four permanent charges, and later gets screened quickly by the mobile ions as time evolves.   Such results on the evolution of $\phi$, $c^{1}$, and $c^{2}$ agree well with previously reported results that are calculated with the PNP equations~\cite{LiuChun2020positivity}. To further understand ion dynamics, Figure~\ref{f:slices} (a) displays the evolution snapshots of cation concentrations at a cross section with $y=-0.5$. Clearly, the cations are attracted by electrostatic interactions to the negative fixed charges. Also, it is interesting to find that the distribution of cations is radially non-monotonic in equilibrium, due to the charge conservation in the system.

%


\begin{figure}[ht]
	\centering
	\includegraphics[scale=0.41]{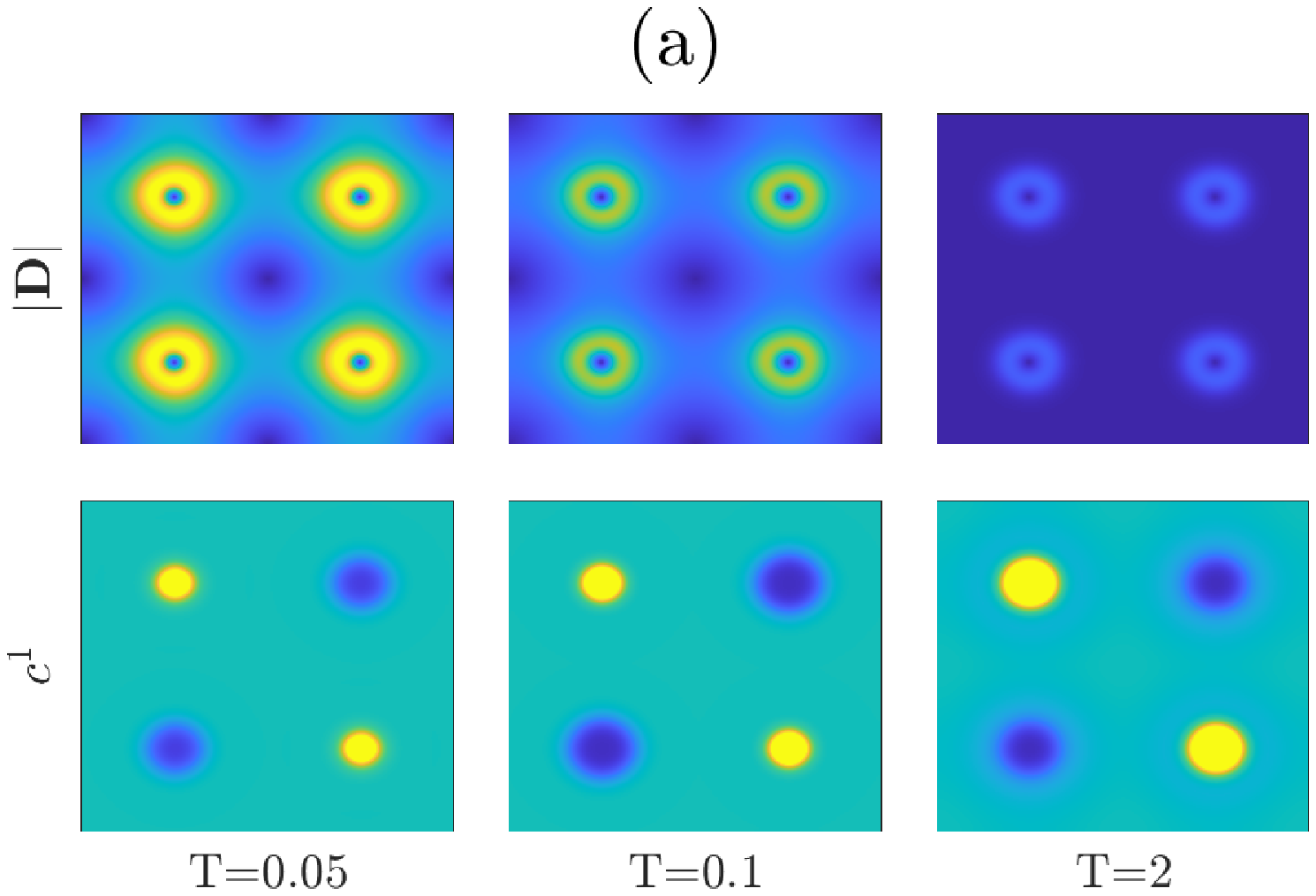} \hspace{-3mm}
	\includegraphics[scale=0.41]{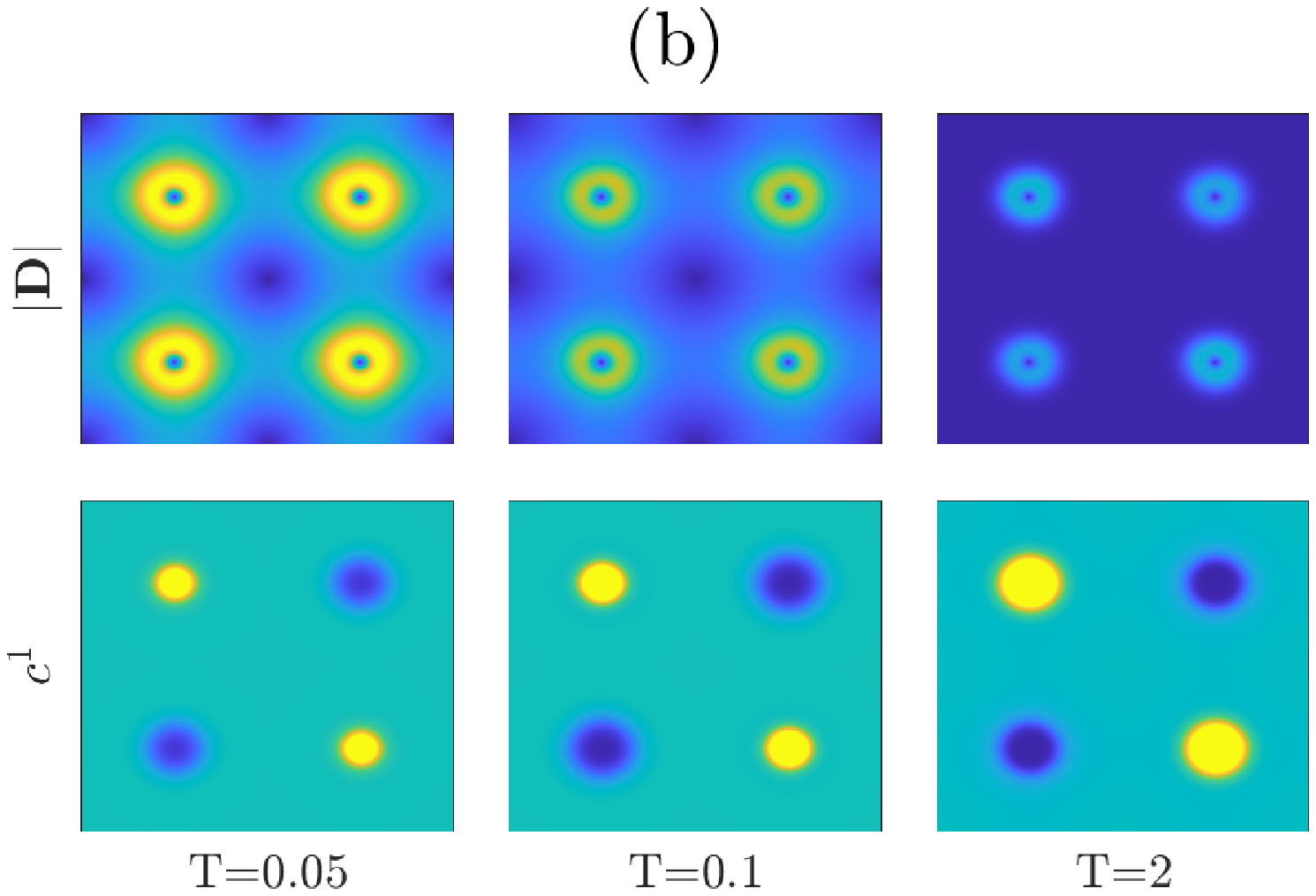}
	\includegraphics[scale=0.41]{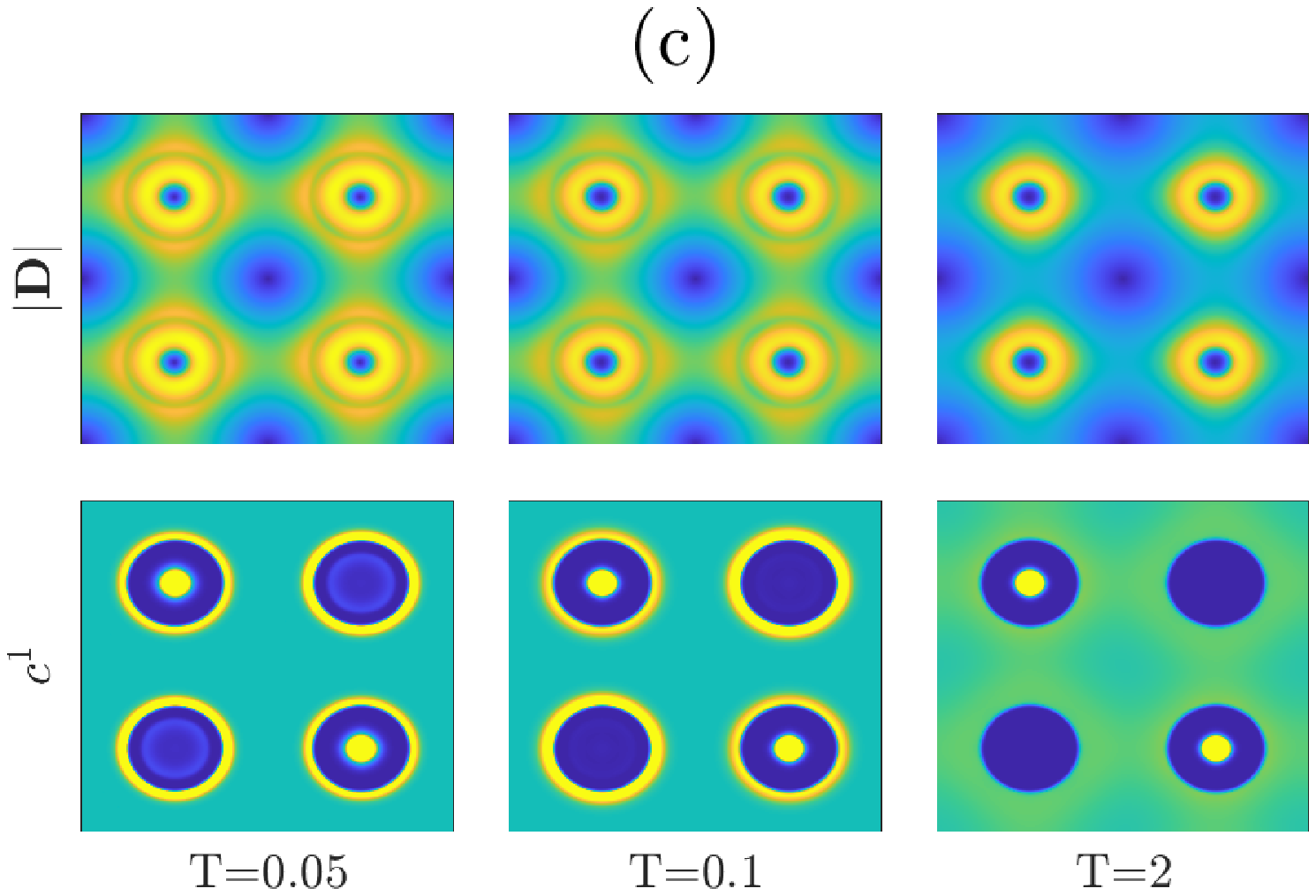} \hspace{-3mm}
	\includegraphics[scale=0.41]{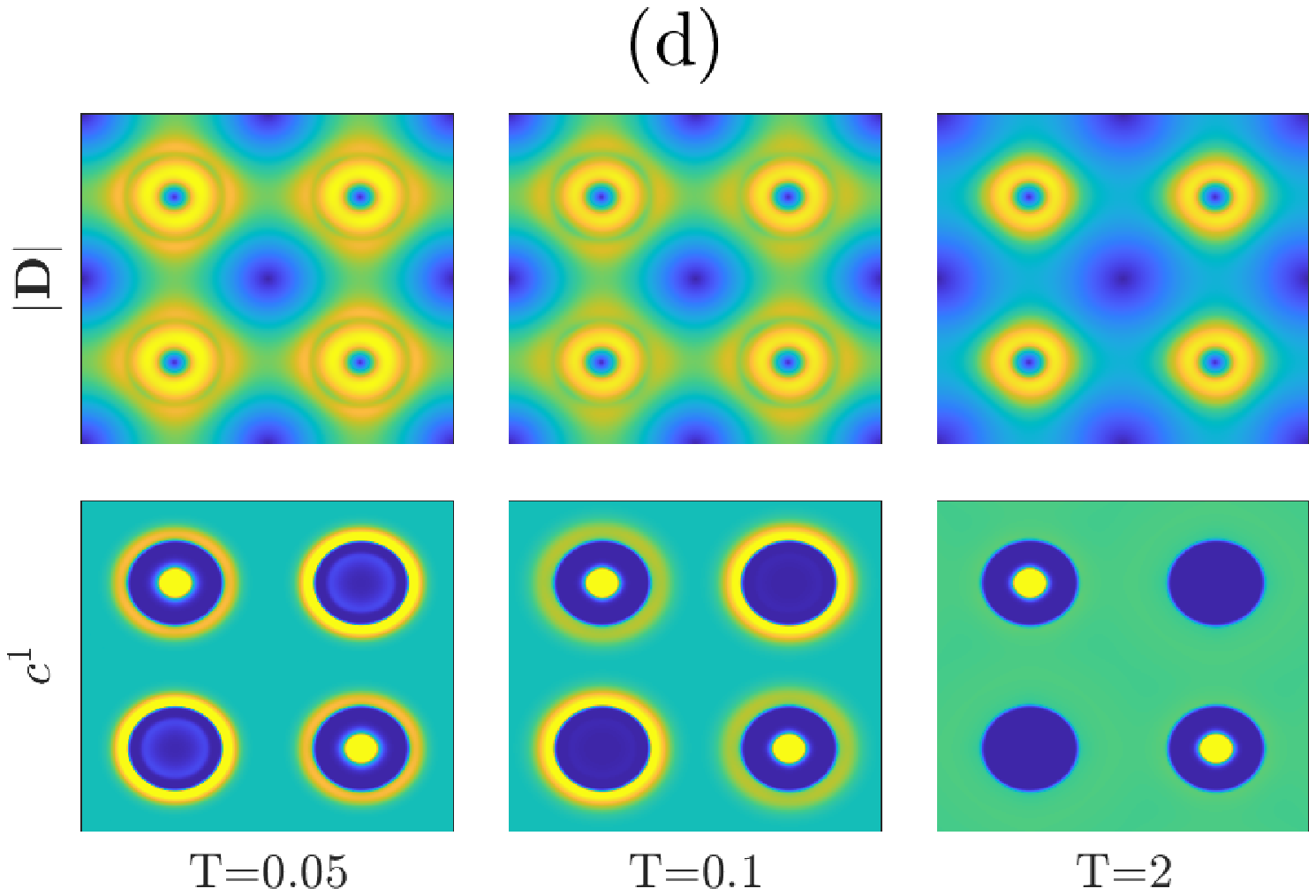}
	\caption{The snapshots of $|\bm{D}|$ and $c^1$ at time $T=0.05$, $T=0.1$, and $T=2$. Panel (a)-(d) correspond to Case 1 - Case 4, respectively. The scales of colorbars are the same as Figure~\ref{f:vem=vew}. }
	\label{f:vem!=vew}
\end{figure}

\begin{figure}[htbp]
	\centering
	\includegraphics[scale=0.43]{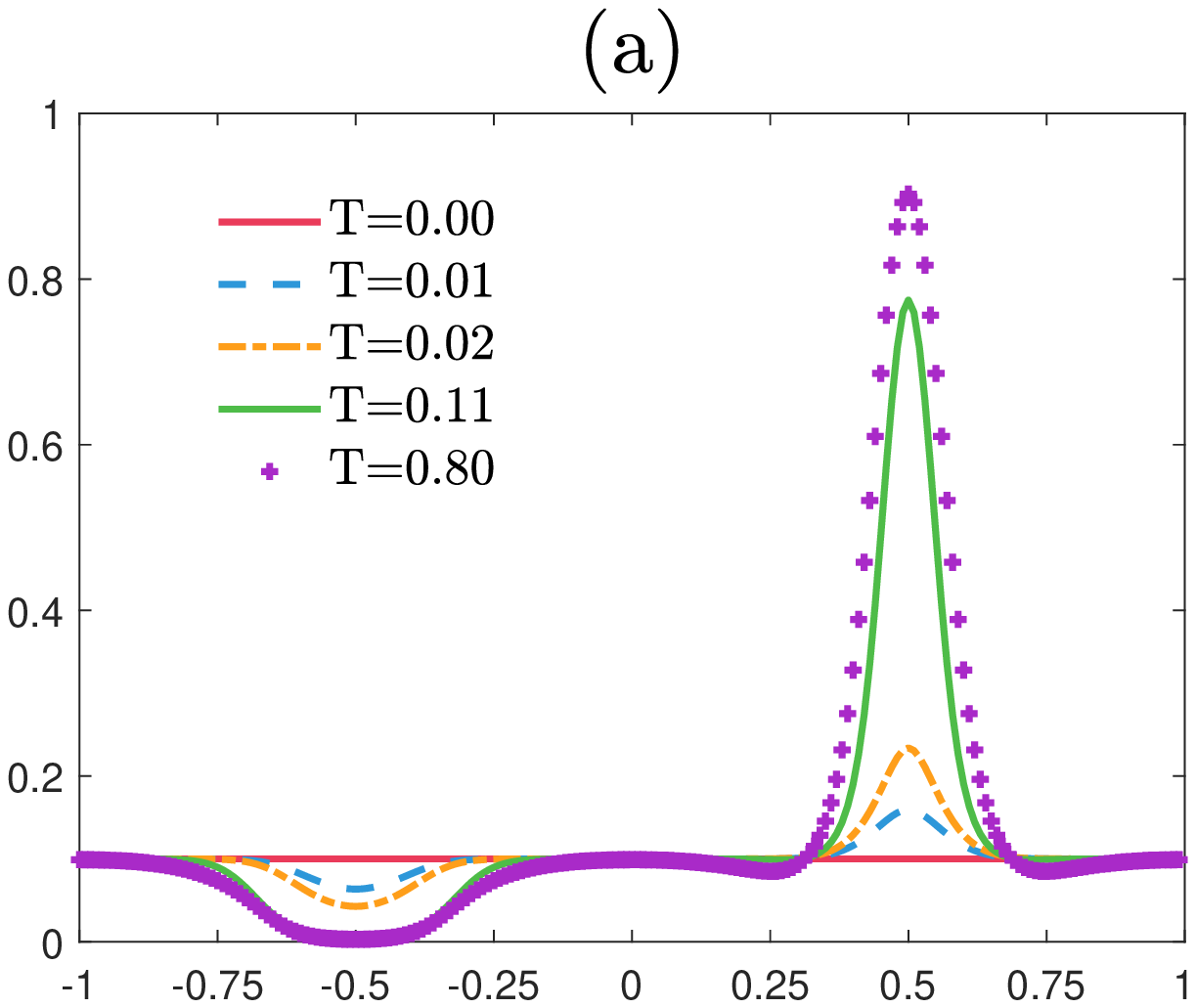} \hspace{-5mm}
	\includegraphics[scale=0.43]{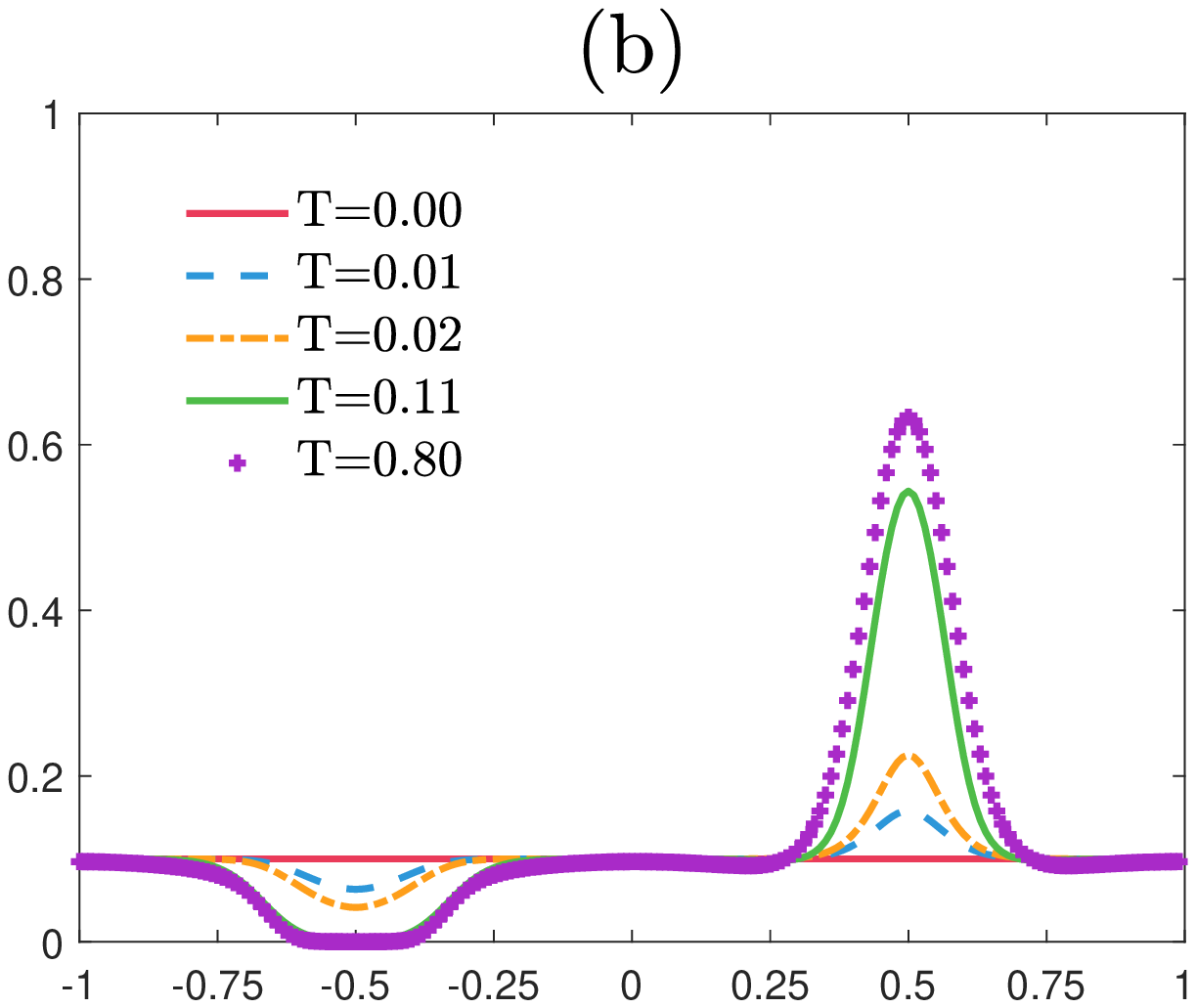}
	\includegraphics[scale=0.43]{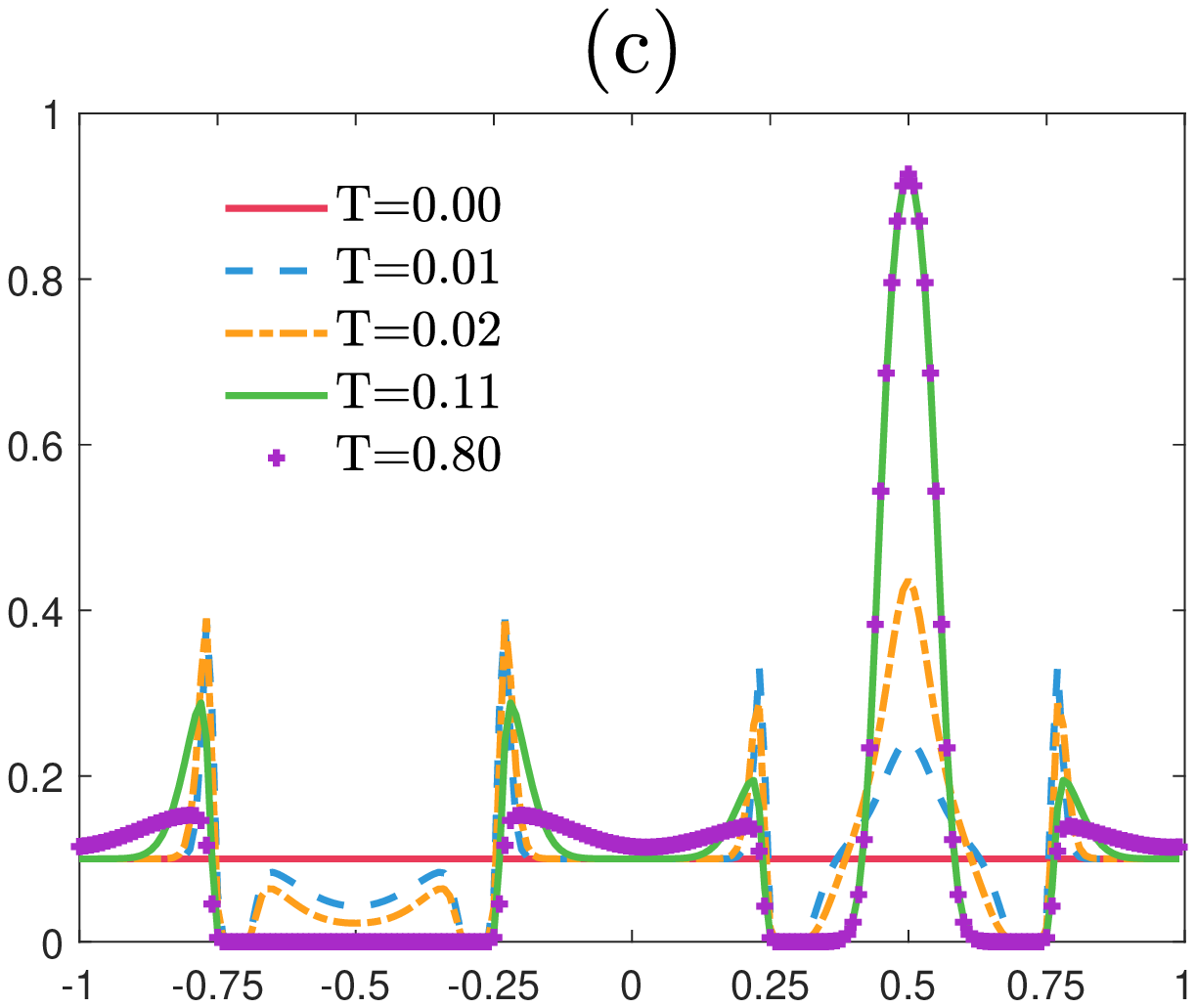}\hspace{-5mm}
	\includegraphics[scale=0.43]{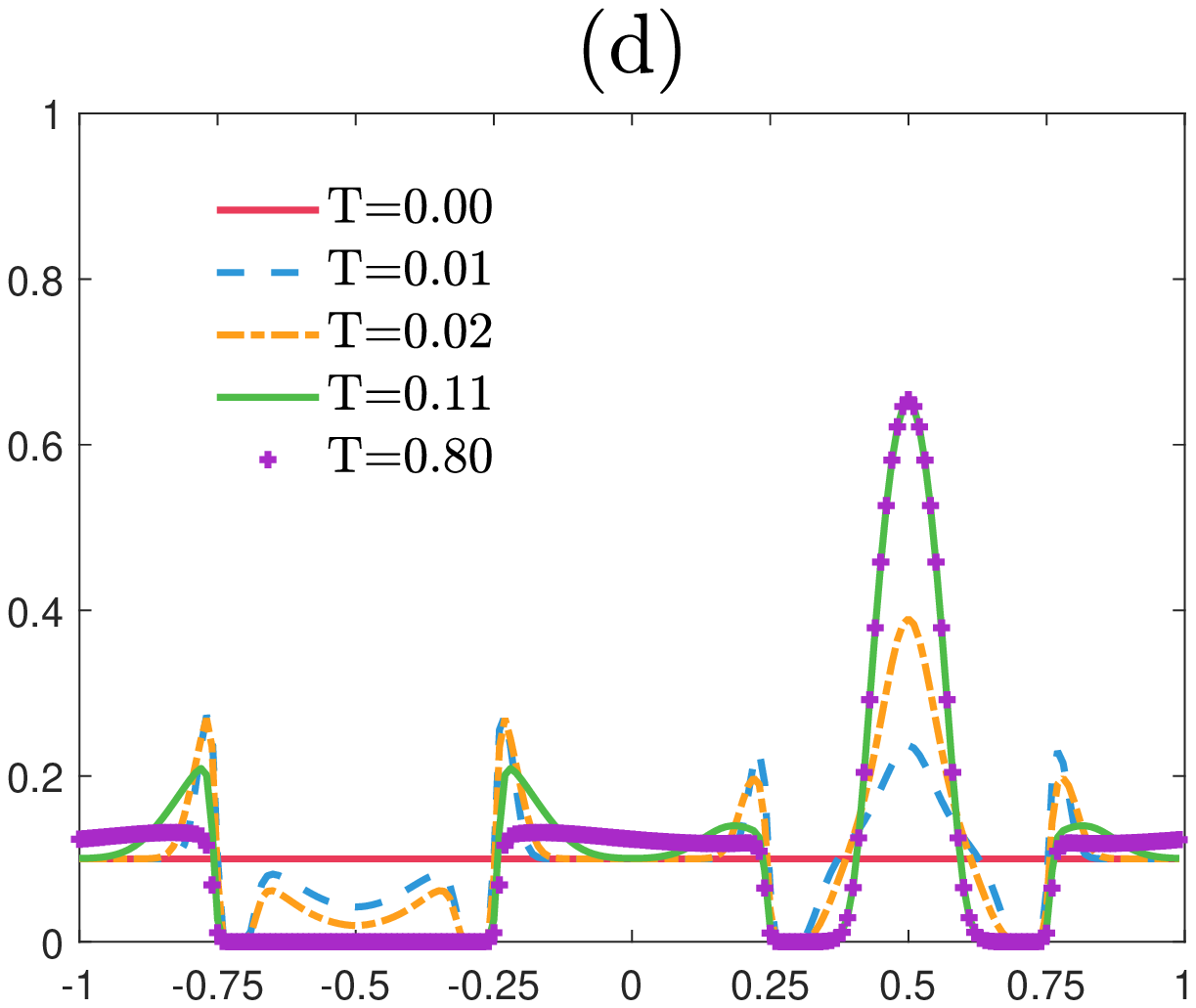}
	\caption{The snapshots of cation concentrations for Case 1 - Case 4  at a cross section with $y=-0.5$.}
	\label{f:slices}
\end{figure}

In Case 2, we still study a uniform dielectric coefficient but with additional ionic steric effects. The results are presented in Figure~\ref{f:vem!=vew} (b) and Figure~\ref{f:slices} (b). Comparing with the results of Case 1, one finds that the counterion concentration peaks at fixed charges lower due to the steric repulsion between ions. With less accumulated ions, the electric displacement is less screened and therefore its magnitude is much stronger than Case 1.

To further assess the Coulomb correlation effects,  we perform numerical simulations with variable dielectric coefficients and Coulomb correlations, as shown in Case 3 and 4. In contrast to previous cases, the dielectric coefficient away from the four permanent charges is $78$ times larger, enhancing the local screening effects, to mimic the big dielectric variation across sharp transition interfaces. The results are displayed in Figure~\ref{f:vem!=vew} (c) - (d) and Figure~\ref{f:slices} (c) - (d). A salient difference between the plots in upper and lower panels is the emergence of the ring structures in both electric displacement and concentration distributions across the sharp dielectric variations. Further detailed comparisons based on Figure~\ref{f:slices} reveal that the Coulomb correlation drives ions to regions with higher dielectric coefficients, forming a marked depletion ring region. The cation concentration has an abrupt transition at the region where the dielectric coefficient has a sharp variation. The concentration snapshots at $T=0.01$ and $T=0.02$ also demonstrate that the Coulomb correlation comes into effect earlier than electrostatic interactions in ion dynamics. Moreover, it is observed from the snapshots that the system relaxes to the steady state faster with the inclusion of Coulomb correlations. Comparing the plots in left and right panels in Figure~\ref{f:slices}, one again can observe that the ionic steric effect lowers concentration peaks by ionic steric repulsions. Overall, the exhibited rich ion dynamics results from the competition among electrostatic interactions, steric effects, and Coulomb correlations.  The magnitude of the electric displacement is relatively stronger with dielectric inhomogeneity. This can be ascribed to the less screening arising from the partly expelled counterions by the additional steric effects and Coulomb correlations, especially in the regions where counterions are almost depleted. This accounts for the eminent ring structures both for the electric displacement and concentration distributions.
\begin{figure}[ht]
	\centering
	\includegraphics[scale=0.5]{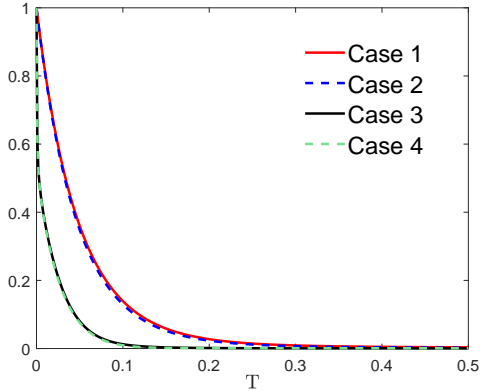}
	\caption{The evolution of normalized discrete energies $\mathcal{F}_h$ for Case 1 to Case 4.}
	\label{f:energy}
\end{figure}
\begin{table}[H]
	\centering
	\begin{tabular}{cccc}
		\hline \hline
		Case 1  & Case 2  & Case 3  & Case 4    \\
		\hline
		1.8139E-3  &  9.5848E-6 & 2.8764E-18 & 1.5079E-24     \\
		\hline \hline
	\end{tabular}
	\caption{The minimum concentration for both cations and anions on the computational grid for the whole time.}
	\label{table-cmin}	
\end{table}
To further understand the impact of various effects on ion dynamics, we investigate the evolution of discrete energies $\mathcal{F}_h$ for Case 1 to Case 4. Figure~\ref{f:energy} shows the discrete energies which are normalized by the difference between the initial and steady-state energies.  From the curves, one can see that the free energy decreases monotonically, as proved in Theorem~\ref{dE/dt}. Also, it is easy to find that the steric effects have minor impact on the evolution of the free energy. Nonetheless, the Coulomb correlations and dielectric inhomogeneity have significant impact on ion dynamics. Being consistent with Figure~\ref{f:slices}, the system reaches the steady state much faster with additional Coulomb correlations. In addition, we also pay attention to the positivity of the numerical solutions to ionic concentrations. Table~\ref{table-cmin} lists the minimum ionic concentration values on the computational grid for the whole time. One can observe that the numerical concentrations remain positive all the time, and the cases with Coulomb correlations and dielectric inhomogeneity have extremely small minimum values due to strong convection caused by the Coulomb correlations and dielectric inhomogeneity.

\section{Conclusions}\label{s:Con}
In this work, we have proposed an MANP framework, based on the ionic concentrations and electric displacement, for the description of charge dynamics. It has been shown that the MANP model {  with a curl-free condition on the electric displacement} is equivalent to the PNP theory. The MANP formulation {is proved to be energy dissipative with respect to} a convex energy functional.
According to the energy dissipation law, the steady state of the MANP formulation has been shown to be the charge conserving PB theory~\cite{Lee_JMP2014,Lee_NonL2010, Wan_PRX2014}. Thus, the MANP formulation offers an energy stable approach to investigate the PB-type theories analytically and numerically. {  To efficiently achieve the curl-free condition, the newly derived MANP formulation has been equipped with a companion local curl-free relaxation algorithm. It has been shown that the algorithm naturally preserves the discrete Gauss's law, has robust convergence, and is of linear computational complexity.} Our MANP formulation also provides a versatile modeling framework to derive models for charge dynamics. Modified MANP models beyond the mean-field theory have been presented to describe cases in which ionic steric effects and Coulomb correlations are not negligible. Results on the ion dynamics with steric effects, Coulomb correlations, and dielectric inhomogeneity have been presented to demonstrate the performance of the proposed MANP model.

It is expected that, with further refinement on the description of ionic interaction details, the proposed MANP model with the local curl-free relaxation algorithm can be extended to the efficient coarse-grained simulations of charge dynamics with fluctuating permanent charges~\cite{TorreEspanol_JCP2011}. {  Also, the MANP model equipped with the local curl-free relaxation algorithm exhibits promising potential in practical applications, especially for the cases involving the variable-coefficient Poisson's equation.} For instance, a Vlasov--Amp\`{e}re model, instead of a  Vlasov--Poisson model, has been applied to particle-in-cell simulations of plasma~\cite{chen2011JCP:VA}. However, only one-dimension simulations have been performed in the work. The proposed MANP model with the local curl-free relaxation algorithm is able to handle high dimensional cases.

\section*{Acknowledgements}

The authors would like to thank the anonymous reviewers for their constructive comments and suggestions, which have helped us greatly improve this work. This work is supported by the CAS AMSS-PolyU Joint Laboratory of Applied Mathematics. Z. Qiao's work is partially supported by the Hong Kong Research Grants Council (RFS grant No. RFS2021-5S03 and GRF grants Nos. 15302919 and 15303121). The work of Z. Xu and Q. Yin was partially supported by  NSFC (grant No. 12071288), Shanghai Science and Technology Commission  (grant Nos. 20JC1414100 and 21JC1403700), the Strategic Priority Research Program of CAS (grant No. XDA25010403) and the HPC center of Shanghai Jiao Tong University. S. Zhou's work is partially supported by the National Natural Science Foundation of China 12171319, Natural Science Foundation of Jiangsu Province (BK20200098), China, and Shanghai Science and Technology Commission (21JC1403700).

\bibliographystyle{plain}
\bibliography{groupbib}

\end{document}